\pdfoutput=1
\documentclass[11pt]{amsart}
\usepackage{mathptmx}
\usepackage[leqno]{amsmath}
\usepackage{amssymb,amsthm} 
\usepackage[parfill]{parskip} 
\usepackage{mathptmx} 
\usepackage{graphicx} 
\usepackage{calc}
\usepackage{showlabels} 
\usepackage{epstopdf}
\usepackage[paperwidth=7in, paperheight=10in, textwidth=4.75in, textheight=8in, includehead=true]{geometry}
\usepackage{url}
\usepackage[dvips,pdfborderstyle={/S/U/W 1}]{hyperref}

\hypersetup{
pdftitle={Terminal Search Multiverse},
pdfauthor={\textcopyright\ S. G. Williamson}
}
\newcommand{\beq}{\begin{equation}}
\newcommand{\eeq}{\end{equation}}
\numberwithin{equation}{section} 
\newtheorem{thm}[equation]{Theorem}

\theoremstyle{remark}

\theoremstyle{definition}
\newtheorem{defn}[equation]{Definition}

\title{Lattice Multiverse Models}
\author{S. Gill Williamson}
\thanks{Department of Computer Science and Engineering, 
University of California San Diego; \url{http://cse.ucsd.edu/~gill}.
}
\date{}                                           
\begin{document}
\maketitle
\begin{abstract}  
Will the cosmological multiverse, when described mathematically, have easily stated properties that are impossible to prove or disprove using mathematical physics? We explore this question by constructing lattice multiverses which exhibit such behavior even though they are much simpler mathematically than any likely cosmological multiverse.

\end{abstract}

\section{Introduction}

We first describe our lattice multiverse models (precise definitions follow). Start with a fixed directed graph $G=(N^k,\Theta)$ (vertex set $N^k$, edge set $\Theta$) where $N$ 
is the set of nonnegative integers and $k\geq 2$. The vertex set $N^k$ of $G$ is the nonnegative $k$ dimensional integral lattice.  If every $(x,y)$ of $\Theta$ satisfies $\max(x) >\max(y)$ where $\max(z)$ is the maximum coordinate value of $z$ then we call  $G$ a {\em downward directed lattice graph}. 
The infinite lattice graph $G$ defines the  set
$\{G_D\mid D\subset N^k\,,\,D\,\,{\rm finite}\}$ of finite vertex induced subgraphs of $G$.  

With each downward directed lattice graph $G$ we associate, in various ways, sets of functions 
$P_G=\{f \mid f: D\rightarrow N, D\subset N^k,\,D\,\,{\rm finite}\}$ (the finite set $D$ is the domain of $f$, and $N$ is the range of $f$).  
Infinite sets of the form ${\bf M}=\{(G_D,f)\mid f\in P_G\,,\,  {\rm domain}(f) = D \}$, will be called lattice ``multiverses"  of $G$ and $P_G$;  the sets $(G_D, f)$ will be the ``universes''  of $\bf M$.

Our use of the terms ``multiverse" and ``universe" in this combinatorial lattice context is inspired by the  analogous but much more complex structures of the same name in cosmology.  The lattice multiverse is a geometric structure for defining the possible lattice universes, $(G_D, f)$, where $G_D$ represents the geometry of the lattice universe and $f$ the things that can be computed about that universe (roughly analogous to the physics of a universe). An example and discussion is given below, see Figure~\ref{fig:univp2}. 

In this paper, we state some basic properties of our elementary lattice multiverses that {\em provably} cannot be proved true or false using the mathematical techniques of physics.  Could the much more complex cosmological multiverses also give rise to conjectured properties provably out of the range of mathematical physics?   Our results suggest  that such a possibility must be considered. 

For the provability results, we rely on the  important work of Harvey Friedman concerning finite functions and large cardinals \cite{hf:nlc} and applications of large cardinals to graph theory \cite{hf:alc}. 

\vskip .25in  
\begin{defn} [\bf Vertex induced subgraph $G_D$]
For any {\em finite} subset $D\subset N^k$ of vertices of $G$, let 
$G_D = (D, \Theta_D)$ be the subgraph of $G$ with vertex set $D$ and edge set 
$\Theta_D =\{(x,y)\mid (x,y)\in \Theta,\, x, y \in D\}$. We call $G_D$ the {\em subgraph of $G$ induced
 by the  vertex set $D$}. 
\label{def:vertexinduced}
\end{defn}

\vskip .25in  
\begin{defn} [\bf Path and terminal path in $G_D$]
A sequence of distinct vertices of $G_D$, $(x_1, x_2, \ldots, x_t)$, is a {\em path}
in $G_D$ if  $t=1$ or if $t>1$ and $(x_i, x_{i+1}) \in \Theta_D,\, i=1, \ldots\,, t-1$.   This path is {\em terminal}
if there is no path of the form  $(x_1, x_2, \ldots, x_t, x_{t+1})$. 
\label{def:path}
\end{defn}

We refer to sets of the form 
$E^k\equiv\times^k E \subset N^k$, $E\subset N$, as {\em {k}-cubes} 
or simply as {\em cubes}.
If $x\in N^k$, then $\min(x)$ is the minimum coordinate value of $x$ and 
$\max(x)$ is  the maximum coordinate value (see discussion of Figure~\ref{fig:univp2}).

\vskip .25in
\begin{defn} [\bf Terminal label function for $G_D$]
Consider a downward directed graph $G=(N^k,\Theta)$ where $N$ 
is the set of nonnegative integers and $k\geq 2$.
For any finite $D\subset N^k$, let $G_D = (D, \Theta_D)$ be the induced subgraph of $G$. 
Define a function $t_D$ on $D$ by
$$t_D(z) = \min(\{\min(x)\mid x\in T_D(z)\} \cup \{\min(z)\})$$
where $T_D(z)$ is the set of all last vertices of terminal paths $(x_1, x_2, \ldots, x_t)$
where $z=x_1$.
We call $t_D$ the {\em terminal label function for $G_D$}.
\label{def:termlabel}
\end{defn}

In words, $t_D(z)$ is gotten by finding all
of the end vertices of terminal paths starting at $z$, taking their minimum coordinate values, throwing in the minimum 
coordinate value of $z$ itself and, finally, taking the minimum of all of these numbers.

\begin{figure}[h]
\begin{center}
\includegraphics[width=3.8in]{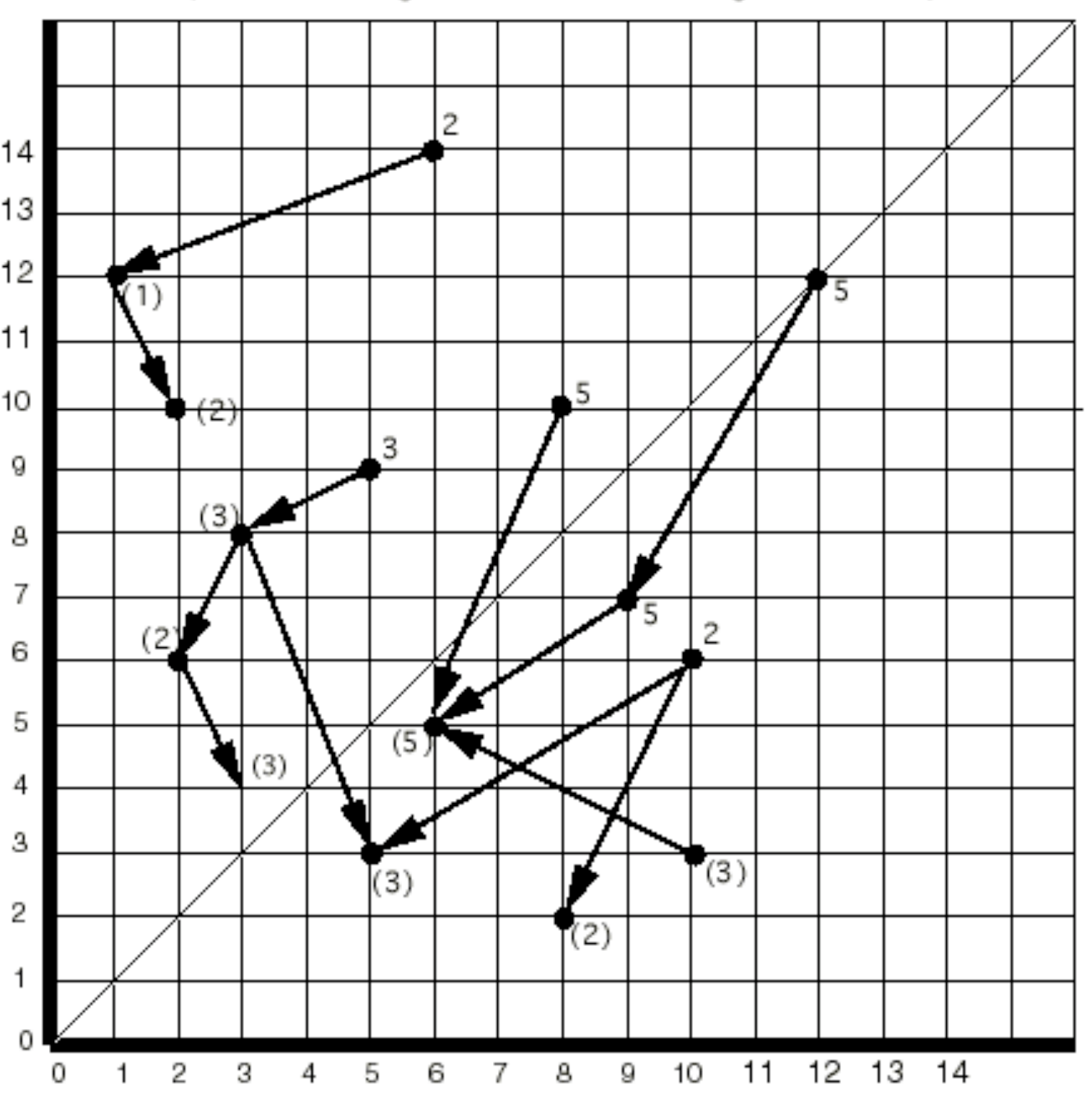}
\caption{Computing $t_D$}\label{fig:univp2}
\end{center}
\end{figure}

Figure~\ref{fig:univp2} shows an example of computing $t_D$ where $D = E \times E$,   
$E = \{0,\ldots, 14\}$.  
The graph $G_D = (D, \Theta_D)$ has $|D|=225$ vertices and $|\Theta_D| = 12$
edges (shown by arrows in Figure~\ref{fig:univp2}). 
Vertices not  on any edge, such as the vertex $(6,10)$, are called {\em isolated} vertices. 
A path  in $G_D$ will be denoted by a sequence of vertices $(x_1, x_2, \ldots, x_t)$, $t\geq 1$. 
For example, $((5,9), (3,8), (2,6))$ is a path:  $x_1=(5,9)$, $x_2 = (3,8)$, $x_3 = (2, 6)$.
Note that the path $((5,9), (3,8), (2,6))$ 
can be extended to $((5,9), (3,8), (2,6), (3,4))$, but this latter path is 
terminal (can't be extended any farther, Definition~\ref{def:path}).  
Note that there is another terminal path shown in Figure~\ref{fig:univp2} 
that starts at $(5,9)$: $((5,9), (3,8), (5,3))$.

As an example of computing $t_D(z)$, look at $z=(5,9)$ in Figure~\ref{fig:univp2} where the value, 
$t_D(z)=3$, of the terminal label function is indicated. 
From Definition~\ref{def:termlabel}, the set $T_D((5,9))=\{(3,4), (5,3)\}$ and $\{\min(x)\mid x\in T_D(z)\} = \{3,3\} = \{3\}$. 
The set $\{\min(z)\} =\{5\}$ and, thus, $\{\min(x)\mid x\in T_D(z)\} \cup \{\min(z)\} = \{3,5\}$ and
$T_D(z)  = \min\{3,5\} = 3$.  If $z$ is isolated,  $t_D(z)=\min(z)$ (for example,
$z=(4,2)$ in the figure is isolated, so $t_D(z)=2$). Such trivial labels are omitted in the figure. 
For $z=(10,3)$,  $T_D(z) = \{(6,5)\}$ so $t_D(z) = \min(z) = 3$. 

\vskip .25in
\begin{defn} [\bf Significant labels]
Let  $t_D$ be the terminal label function for $G_D$ and let $S\subset D$. The set $\{t_D(z)\mid z\in S\,,\, t_D(z) < \min(z)\}$ is the set of 
$t_D$--{\em significant labels} of $S$ in $D$. 
 \label{def:siglabels}
\end{defn}

Referring to Figure~\ref{fig:univp2} with $S=D$, 
$\{(5,9), (6,14), (8,10), (9,7), (10,6), (12,12)\}$ are vertices with significant labels,
and the set of significant labels is
$\{t_D(z)\mid z\in S\,,\, t_D(z) < \min(z)\} = \{2, 3, 5\}.$ The terminology comes from the ``significance" of these number
with respect to order type equivalence classes and the concept of regressive regularity (e.g., Theorem~\ref{thm:P2}). The set of significant labels also occurs in certain studies of lattice embeddings of posets~\cite{jg:pos}.

In the next section, we study the set of significant labels. 

\section{Lattice Multiverse TL}

We start with a definition and related theorem that we state without proof.
\vskip .25in
\begin{defn} [\bf full, reflexive, jump-free] 
Let $Q$ denote a collection of functions whose domains are finite subsets of $N^k$ and ranges are subsets of $N$. 

\begin{enumerate}
\item {\bf full:}  We say that $Q$ is a {\em full} family of functions on $N^k$ if for every finite subset 
$D\subset N^k$ there is at least one function $f$ in $Q$ whose domain is $D$.

\item{\bf reflexive:} We say that  $Q$ is a {\em reflexive} family of functions on $N^k$ if for 
every $f$ in $Q$ and for each $x\in D$, $D$ the domain of $f$, $f(x)$ is a coordinate of some $y$ in $D$.

\item{\bf jump-free:} For $D\subset N^k$ and $x\in D$ define $D_x = \{z\mid z\in D,\, \max(z) < \max(x)\}$. 
Suppose that for all $f_A$ and $f_B$  in $Q$, where $f_A$ has domain $A$ and $f_B$ has domain $B$,  the conditions
 $x\in A\cap B$, $A_x \subset B_x$, and $f_A(y) = f_B(y)$ for all $y\in A_x$ imply that 
$f_A(x) \geq f_B(x)$.  Then $Q$ will be called a {\em jump-free} family of functions on $N^k$.  
\end{enumerate}
\label{def:fullrefjf}
\end{defn}

Figure~\ref{fig:jumpfree} may be helpful in thinking about the jump-free condition ($k=2$). 
The square shown in the figure has sides of length $\max(x)$. 
The set $A_x$ is the intersection of the set $A$ with the set of lattice points
interior to the square.  The set $B_x$ is this same intersection for the set $B$.
\begin{figure}[h]
\begin{center}
\includegraphics[width=3in]{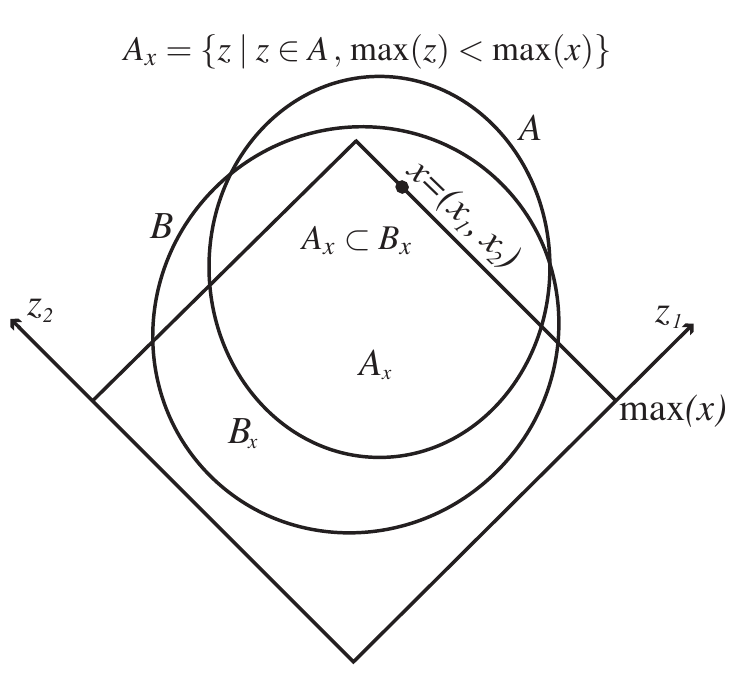}
\caption{Jump-free ``light cone''}\label{fig:jumpfree}
\end{center}
\end{figure}

To prove our main result, we use a theorem of 
\href{http://www.math.ohio-state.edu/~friedman} {Harvey Friedman} called the ``jump-free theorem,'' 
Theorem~\ref{thm:jumpfree}. 
The jump-free theorem is proved and shown to be independent of the ZFC (Zermelo, Fraenkel, Choice) 
 axioms of mathematics in Section 2 of \cite{hf:alc},
``Applications of Large Cardinals to Graph Theory,'' October 23, 1997, No. 11 of  
\href{http://www.math.ohio-state.edu/~friedman//manuscripts.html}{Preprints, Drafts, and Abstracts}.
The proof uses results from \cite{hf:nlc}.

\vskip .25in

\begin{thm}[\bf Friedman's jump-free theorem] 
Let $Q$ denote a full, reflexive, and jump-free family of functions on $N^k$ (Definition~\ref{def:fullrefjf}).
Given any integer $p>0$, there is a finite $D\subset N^k$ and a subset $S=E^k\subset D$ with $|E|=p$ 
such that for some $f \in Q$ with domain $D$, the set
$\{f(z)\mid z\in S, f(z)<\min(z)\}$ has at most cardinality $k^k$.
 \label{thm:jumpfree}
\end{thm}
{\em Technical Note:} The function $f$ of the jump-free theorem can be chosen such that for 
each order 
type~\footnote{Two k-tuples, $x=(x_1,\ldots,x_k)$ and $y=(y_1,\ldots,y_k)$, have the same order type if 
$\{(i,j)\mid x_i < x_j\} =  \{(i,j)\mid y_i < y_j\}$ and  $\{(i,j)\mid x_i = x_j\} =  \{(i,j)\mid y_i =y_j\}.$}
 $\omega$ of $k$-tuples, either $f(x) \geq \min(x)$ for all $x\in E^k$ where $x$ is
of type $\omega$ or $f(x) = f(y) < \min(E)$ for all $x\in E^k$ and $y\in E^k$, $x$ and $y$ of order type $\omega$.
We call  such a function {\em regressively regular} over $E^k$.  Note that for $k\geq 2$, the number of order type 
equivalence classes is always strictly less than $k^k$.

\vskip .25in
\begin{defn} [\bf Multiverse TL -- terminal label multiverse]
Let $G=(N^k,\Theta)$ be a downward directed graph where $N$ is the nonnegative integers. 
Define $\bf M_{TL}$ to be the set  $\{(G_D, t_D) \mid D\subset N^k, \,\,D\,\, {\rm finite}\}$
where  $t_D$ is the terminal label function of  the induced subgraph $G_D$. 
We call $\bf M_{TL}$ a k-dimensional multiverse of type~TL. We refer to the pairs $(G_D, t_D)$ as the 
universes of $\bf M_{TL}$.
\label{def:multversep2}
\end{defn}

We now use Friedman's jump-free theorem to prove a basic structure theorem for Multiverse~TL. Intuitively, this theorem (Theorem~\ref{thm:P2}) states that for any specified cube size, no matter how large, there is a universe of Multiverse~TL that contains a cube of that size with certain special properties.

\vskip .25in
\begin{thm}[\bf Multiverse TL] 
Let $\bf M_{TL}$ be a k-dimensional lattice multiverse of type~TL and let $p$ be any positive integer. Then there
is a universe $(G_D, t_D)$ of $\bf M_{TL}$ and a subset $E\subset N$ with $|E|=p$ and $S=E^k\subset D$ such that 
the set of significant labels $\{t_D(z)\mid z\in S, t_D(z) < \min(z)\}$ has size at most $k^k$. In fact, $t_D$ is
regressively regular over $E^k$.
\label{thm:P2}
\end{thm}

\begin{proof}
Recall that $t_D$ is the terminal labeling function of  the induced subgraph 
$G_D = (D, \Theta_D)$  of the graph  $G=(N^k,\Theta)$.
We apply Theorem~\ref{thm:jumpfree} to a ``relaxed'' version, $\hat {t}_D$, of $t_D$ defined by
$\hat {t}_D(z) = \max(z)$ if $(z)$ is a terminal path in $G_D$ and $\hat {t}_D(z) = t_D(z)$  
otherwise.~\footnote{This clever idea is due to Friedman~\cite{hf:alc}.}
If $(z)$ is a terminal path in $G_D$ then $\hat {t}_D(z) = \max(z)$ by definition, and
if $(z)$ is {\em not} a terminal path in $G_D$, the downward condition implies that
$\hat {t}_D(z) = t_D(z) < \max(z)$. 
Thus,   $\hat {t}_D(z) \leq \max(z)$ with equality if and only if $(z)$ is terminal.

Let $Q$ denote the collection of functions $\hat {t}_D$ as $D$ ranges over all finite subsets of $N^k$.  
We will show that $Q$ is full, reflexive, and jump-free (Definition~\ref{def:fullrefjf}).
Full and reflexive are obvious from the definition of $\hat {t}_D$.
We want to show that for all $\hat {t}_A$ and $\hat {t}_B$  in $Q$ the conditions
 $x\in A\cap B$, $A_x \subset B_x$, and $\hat {t}_A(y) = \hat {t}_B(y)$ for all $y\in A_x$ imply that 
$\hat {t}_A(x) \geq \hat {t}_B(x)$.

Suppose  that $(x)$ is terminal in $G_A$.  Then $\hat{t}_A(x) = \max(x)\geq \hat{t}_B(x)$ from our
observations above.

Suppose that $(x)$ is not terminal in $G_A$.  Then $\hat{t}_A(x) = t_A(x)$ by definition of $\hat {t}_A$.  
From the definition of $t_A(x)$, there is a path, $x=x_1, x_2, \ldots, x_t$, with $t> 1$ such that
$t_A(x)=\min(x_t)$ and $(x_t)$ is terminal in $G_A$.  Thus,  $\hat {t}_A(x_t) = \max(x_t)$.  Our basic assumption
is that  $\hat {t}_A(y) = \hat{t}_B(y)$ for all $y\in A_x$ and hence for $y=x_t$. 
Thus, $\hat{t}_B(x_t)=\max(x_t)$ and hence, from our discussion above, $(x_t)$ is also terminal in $G_B$.
Since $A_x \subset B_x$, the path $x=x_1, x_2, \ldots, x_t$ with $t> 1$ is also a terminal path in $G_B$.
Thus, $x_t \in T_B(x)$ (Definition~\ref{def:termlabel}) and $t_B(x)\leq \min(x_t) = t_A(x)$.
Since $(x)$ is not terminal in either $G_A$ or $G_B$, $\hat {t}_B(x) =t_B(x)\leq \min(x_t) = t_A(x) = \hat {t}_A(x)$ 
which completes the proof that $Q$ is jump-free.

From Theorem~\ref{thm:jumpfree}, given any integer $p>0$, there is a finite $D\subset N^k$ and a 
subset $E^k\subset D$ with $|E|=p$ such that, for some $\hat{t}_D \in Q$, the set
$\{\hat{t}_D(z)\mid z\in S,\, \hat{t}_D(z)<\min(z)\}$ has at most cardinality $k^k$.
In fact, $\hat{t}_D(z)$ is regressively regular on $E^k$ (See Theorem~\ref{thm:jumpfree}, technical note).

Finally, we note that if $t_D(z)<\min(z)$ then $(z)$ is not a terminal path in $G_D$ and hence 
$\hat{t}_D(z) = t_D(z) < \min(z)$.  
Thus, $\{t_D(z)\mid z\in S,\, t_D(z)<\min(z)\}$ has at most cardinality $k^k$ also.  
In fact, $\hat{t}_D(z)$ is regressively regular on $E^k$ implies that $t_D(z)$ is regressively regular on $E^k$.

To see this latter point, suppose that  $\hat{t}_D(z) \geq \min(z)$ for all $z\in E^k$ of order type $\omega$.
If $(z)$ is terminal, $t_D(z) = \min(z)$. If $(z)$ is not terminal, $t_D(z) = \hat t_D(z) \geq \min(z)$. Thus,
$\hat{t}_D(z) \geq \min(z)$ for all $z\in E^k$ of order type $\omega$ implies 
$t_D(z) \geq \min(z)$ for all $z\in E^k$ of order type $\omega$.  

Now suppose that
for all $z,w \in E^k$ of order type $\omega$,  $\hat{t}_D(z) = \hat{t}_D(w) < \min(E)$.
This inequality implies that  $\hat{t}_D(z) < \min(z)$ and $\hat{t}_D(w) < \min(w)$ and thus $(z)$ and $(w)$
 are not terminal.  
 Hence, $t_D(z) = \hat{t}_D(z) = \hat{t}_D(w) = t_D(w) < \min(E)$.  
\end{proof}

{\bf Summary:} We have proved that given an arbitrarily large cube, there is some universe
$(G_D, t_D)$ of $\bf M_{TL}$ for which the ``physics," $t_D$, has a simple structure over a cube of that size.
To prove this large-cube property, we have  used a theorem independent of ZFC.  We do not know if this
large-cube property can be proved in ZFC. The mathematical techniques of physics lie within the 
ZFC axiomatic system.

\section{Lattice Multiverse SL}

We now consider a class of multiverses where the ``physics" is more complicated than in Multiverse TL. 
For us, this means that the label function  is more complicated than $t_D$. 
Again, we consider a fixed downward directed graph $G=(N^k,\Theta)$ where $N$ 
is the set of nonnegative integers.

\vskip .25in
\begin{defn} [{\bf Partial selection}]
A function $F$ with domain a subset of $X$ and range a subset of $Y$ will be called a {\em partial function}
from $X$ to $Y$ (denoted by $F: X\rightarrow Y$).  If $z\in X$ but $z$ is not in the domain of $F$, we say 
$F$ is {\em not defined} at $z$.
A  partial function 
$F: (N^k \times N)^r \rightarrow N$
will be called a {\em partial selection} function if whenever 
$F((y_1,n_1), (y_2,n_2), \ldots (y_r,n_r))$ is defined we have 
$F((y_1,n_1), (y_2,n_2), \ldots (y_r,n_r)) = n_i$ for some $1\leq i \leq r$.
\label{def:partselect}
\end{defn}
 
For $x\in N^k$, let $G^x = \{y \mid (x,y)\in \Theta\}$.  $G^x$ is the set of
vertices {\em adjacent} to $x$ in $G$.  

For $x\in N^k$ a vertex of $G$ and $r\geq 1$, let $F^x_r : (G^x \times N)^r \rightarrow N$ denote a partial function.
Let ${\mathcal{F}}_G = \{F^x_r\mid x\in N^k\,,\, r\geq 1\}$ be the set of partial functions for $G$. 

Let $G^x_D$ denote the set of vertices of $G_D$ that are {\em adjacent} to $x$ in $G_D$.

\vskip .25in
\begin{defn} [{\bf Selection labeling function $s_D$ for $G_D$}]
For $D\subset N^k$, $x\in D$, we define $s_D(x)$ by induction on $\max(x)$. For each $x\in N^k$, let
$$
	\Phi^D_x = \{ F_r^x((y_1,n_1), (y_2,n_2), \ldots (y_r,n_r))\mid y_1, \ldots,  y_r \in G^x_D,\, r\geq 1\,,\,F_r^x\in \mathcal{F}_G\} 
$$
be the set of defined values of $F_r^x$ where $s_D(y_1)=n_1$, \ldots, $s_D(y_r) = n_r$.
If $\Phi^D_x\neq \emptyset$, let $s_D(x)$ be the minimum over $\Phi^D_x$; otherwise, let $s_D(x) = \min(x)$.~\footnote
{
Our students called this the ``committee labeling function.''  The graph $G_D$ describes the structure of an organization.  The function $\Phi^D_x = F_r^x((y_1,n_1), (y_2,n_2), \ldots (y_r,n_r))=n_i$ represents the 
committee members $y_j$ with individual reports $n_j$.  The boss, $x$ (an ex officio member), makes a decision $n_i$ after taking into account
the committee and its inputs.
}

\label{def:chanlabel}
\end{defn}

\vskip .25in
\begin{defn} [\bf Multiverse SL -- selection label multiverse]
Let $G=(N^k,\Theta)$ be a downward directed graph where $N$ is the nonnegative integers. 
Define $\bf M_{SL}$ to be the set of universes $\{(G_D, s_D) \mid D\subset N^k, \,\,D\,\, {\rm finite}\}$
where  $s_D$ is the selection labeling function of  the induced subgraph $G_D$. 
We call $\bf M_{SL}$ a k-dimensional multiverse of type~SL and $(G_D, s_D)$ a universe of type~SL.\label{def:multversep4}
\end{defn}

The following theorem (Theorem~\ref{thm:P4}) asserts that the ``large-cube" property is  valid for Multiverse~SL.  

\vskip .25in
\begin{thm}[\bf Multiverse SL] 
Let $\bf M_{SL}$ be a k-dimensional multiverse of type~SL and let $p$ be any positive integer. Then there
is a universe $(G_D, s_D)$ of $\,{\bf M_{SL}}$ and a subset $E\subset N$ with $|E|=p$ and $S=E^k\subset D$ such that 
the set of $s_D$--significant labels $\{s_D(z)\mid z\in S, s_D(z) < \min(z)\}$ has size at most $k^k$. In fact, $s_D$ is
regressively regular over $E^k$.
\label{thm:P4}
\end{thm}
\begin{proof}
Define $\hat s_D(x) = s_D(x)$ if $\Phi^D_x \neq \emptyset$.  Otherwise, define $\hat s_D(x) = \max(x)$.  
Induction on $\max(x)$ shows that $\hat s_D(x) \leq \max(x)$ with equality if and only if $\Phi^D_x = \emptyset$. 

Let $Q$ denote the collection of functions $\hat {s}_D$ as $D$ ranges over all finite subsets of $N^k$.  
We will show that $Q$ is full, reflexive, and jump-free (Definition~\ref{def:fullrefjf}).
Full and reflexive are obvious from the definition of $\hat {s}_D$.
We want to show that for all $\hat {s}_A$ and $\hat {s}_B$  in $Q$ the conditions
 $x\in A\cap B$, $A_x \subset B_x$, and $\hat {s}_A(y) = \hat {s}_B(y)$ for all $y\in A_x$ imply that 
$\hat {s}_A(x) \geq \hat {s}_B(x)$.

If $\Phi^A_x = \emptyset$, then $\hat s_A(x) = \max(x)$ and thus $\hat {s}_A(x) \geq \hat {s}_B(x)$.
Suppose  $\Phi^A_x \neq \emptyset$. The conditions $x\in A\cap B$ and $A_x \subset B_x$ imply 
that $G^x_A \subset G^x_B$ and hence, using $s_B(y_i) = s_A(y_i)$, $1\leq i \leq r$, that 
$$\Phi^A_x = \{ F_r^x((y_1,s_A(y_1)), \ldots (y_r,s_A(y_r)))\mid y_1, \ldots,  y_r \in G^x_A\,,r\geq 1\}$$
equals
$$\{ F_r^x((y_1,s_B(y_1)), \ldots (y_r,s_B(y_r)))\mid y_1, \ldots,  y_r \in G^x_A\,,r\geq 1\}$$
which is contained in
$$\Phi^B_x = \{ F_r^x((y_1,s_B(y_1)), \ldots (y_r,s_B(y_r)))\mid y_1, \ldots,  y_r \in G^x_B\,,r\geq 1\}.$$
Thus, we have $\emptyset \neq \Phi^A_x \subset \Phi^B_x$ and hence $s_A(x) = \min(\Phi^A_x)\geq \min(\Phi^B_x) = s_B(x)$.
Since both $\Phi^A_x$ and $\Phi^B_x$ are nonempty, we have 
$\hat s_A(x) = s_A(x) \geq s_B(x) = \hat s_B(x)$. This  shows that  $Q=\{\hat s_D: D\subset N^k, D\,\, {\rm finite}\}$ is jump-free.  
From Theorem~\ref{thm:jumpfree} and the {\em Technical Note}, given any integer $p>0$, 
there is a finite $D\subset N^k$ and a 
subset $E^k\subset D$ with $|E|=p$ such that, for some $\hat{s}_D \in Q$, the set
$\{\hat{s}_D(z)\mid z\in S,\, \hat{s}_D(z)<\min(z)\}$ has at most cardinality $k^k$.  In fact, $\hat s_D$ is regressively 
regular over $E^k$.  Finally, we must show that $s_D$ itself satisfies the conditions just stated for $\hat s_D$.  

To see this latter point, suppose that  $\hat{s}_D(x) \geq \min(x)$ for all $x\in E^k$ of order type $\omega$.
If $\Phi_x^D = \emptyset$ then $s_D(x) = \min(x)$. If $\Phi_x^D \neq \emptyset$ then $s_D(x) = \hat s_D(x) \geq \min(x)$. Thus, $s_D(x) \geq \min(x)$ for all $x\in E^k$ of order type $\omega$.
 
Now suppose that
for all $z,w \in E^k$ of order type $\omega$,  $\hat{s}_D(z) = \hat{s}_D(w) < \min(E)$.
This inequality implies that  $\Phi_z^D \neq \emptyset$ and $\Phi_w^D \neq \emptyset$ 
and thus $s_D(z) = \hat{s}_D(z) = \hat{s}_D(w) = s_D(w) < \min(E)$.
 
Thus, the set $s_D$ is regressively regular over $E^k$.   
And $\{{s}_D(z)\mid z\in S,\, {s}_D(z)<\min(z)\}$ has at most cardinality $k^k$.

\end{proof}

\begin{figure}[h]
\begin{center}
\includegraphics[width=5in]{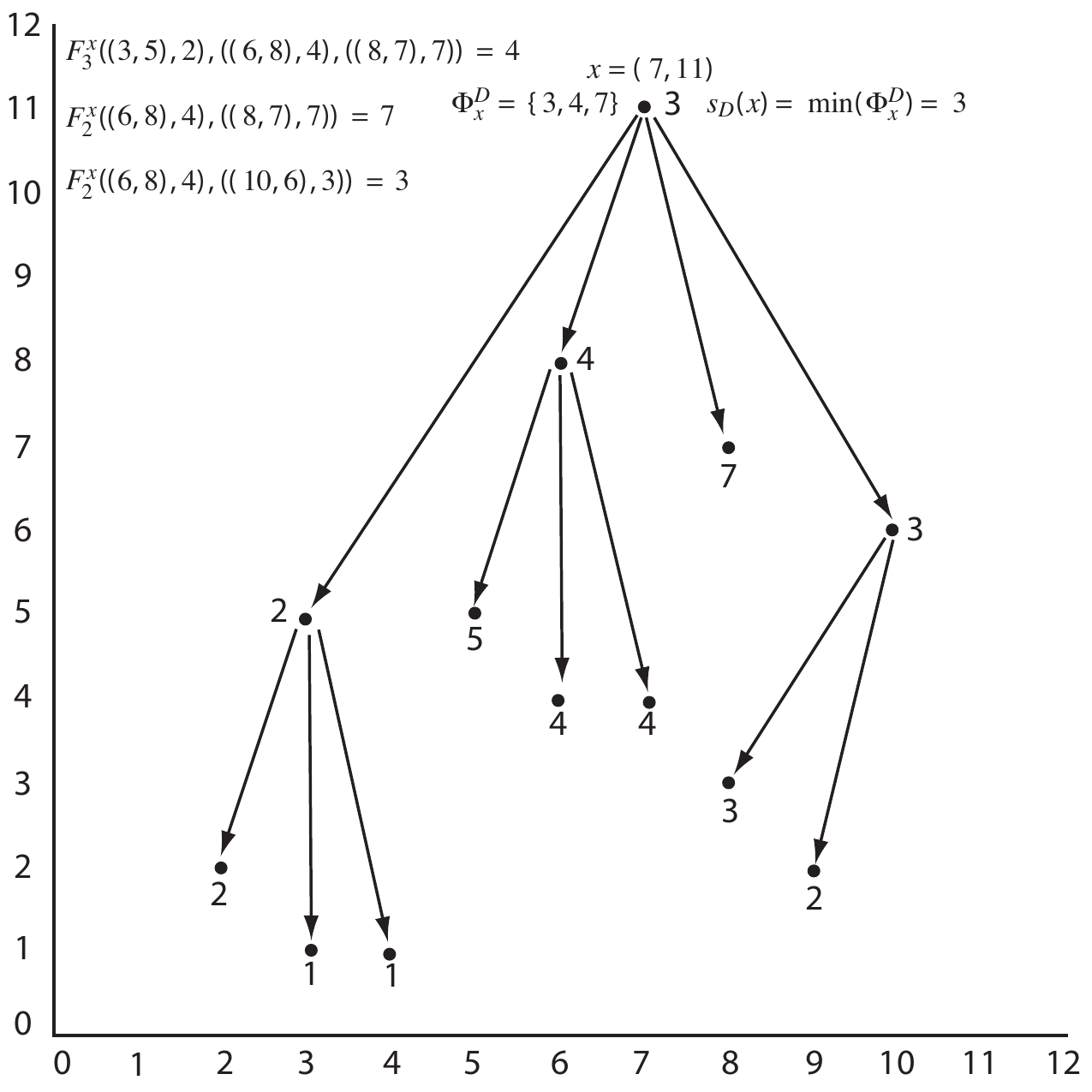}
\caption{Computing $s_D(x)$ in a universe of type SL}\label{fig:univS}
\end{center}
\end{figure}

An example of a universe of type SL is given in Figure~\ref{fig:univS} where we show how $s_D(x)$ was 
computed for $x=(7,11)$ using Definition~\ref{def:chanlabel}.

Note that the directed graph structure, $G=(N^k,\Theta)$, which parameterizes Theorem~\ref{thm:P4} can be thought of as the
``geometry"  of  Multiverse~SL.  This geometry has intuitive value in constructing examples 
but is a nuisance in proving independence. If we take $\Theta = \{(x,y)\mid \max(x) > \max(y)\}$
to be the maximal possible set of edges then the graph structure can be subsumed in the partial selection functions.
We refer to the resulting ``streamlined" multiverse $\bf  M_{SL}$ as Multiverse~HF.  
Theorem~\ref{thm:F} below then follows from Theorem~\ref{thm:P4}.  

\vskip .25in
\begin{thm}[\bf Multiverse HF]
Let $\bf M$ be the k-dimensional multiverse of type~S where $\Theta = \{(x,y)\mid \max(x) > \max(y)\}$.
Let $p$ be any positive integer. Then there
is a universe $(G_D, s_D)$ of $\bf M$ and a subset $E\subset N$ with $|E|=p$ and $S=E^k\subset D$ such that 
the set of significant labels $\{s_D(z)\mid z\in S, s_D(z) < \min(z)\}$ has size at most $k^k$. In fact, $s_D$ is
regressively regular over $E^k$.
\label{thm:F}
\end{thm}

A special case of Theorem~\ref{thm:F} above (where the parameter $r$ is fixed in defining the $s_D$)
is equivalent to Theorem~4.4 of \cite{hf:alc}. 
Theorem 4.4 has been shown by Friedman to be independent of the ZFC axioms of mathematics (see Theorem~4.4 through Theorem~4.15 \cite{hf:alc} and  Lemma~5.3, page ~840, \cite{hf:nlc}). 

{\bf Summary:} We have proved that given an arbitrarily large cube, there is some universe
$(G_D, s_D)$ of $\bf M_{SL}$ for which the ``physics," $s_D$, has a simple structure over a cube of that size.
To prove this large-cube property, we have  used a theorem independent of ZFC.  No proof using just the ZFC axioms is possible. All  of the mathematical techniques of physics lie within the ZFC axiomatic system.
\section{Final Remarks}

For a summary of key ideas involving multiverses, see Linde \cite{al:iuv} and 
Tegmark \cite{mt:mhy}. Tegmark describes four stages of a possible multiverse  theory and discusses the mathematical and physical implications  of  each.  For a well written and thoughtful presentation of the multiverse concept in cosmology, see Sean Carroll \cite{sc:eth}.

Could foundational issues analogous to our assertions about large cubes occur in the study of cosmological multiverses?  The set theoretic techniques we use in this paper are fairly new and not  known to most mathematicians and physicists, but a growing body of useful ZFC--independent theorems like the 
jump-free theorem, Theorem~\ref{thm:jumpfree}, are being added to the set theoretic toolbox.
The existence of structures in a cosmological multiverse corresponding to  our lattice multiverse cubes (and requiring ZFC--independent proofs) could be a subtle artifact of the mathematics, physics, or geometry of the multiverse. 

{\bf Acknowledgments:}  The author thanks  Professors Jeff Remmel and Sam Buss (University of California San Diego, Department of Mathematics) and Professor Rod Canfield (University of Georgia, Department of Computer Science) for their helpful comments and suggestions.


\bibliographystyle{alpha}
\bibliography{multiverse}

\begin{thebibliography}{RW99}

\bibitem[Car10]{sc:eth}
Sean Carroll.
\newblock {\em From Eternity to Here}.
\newblock Dutton, New York, 2010.

\bibitem[Fri97]{hf:alc}
Harvey Friedman.
\newblock Applications of large cardinals to graph theory.
\newblock Technical report, Department of Mathematics, Ohio State University,
  1997.

\bibitem[Fri98]{hf:nlc}
Harvey Friedman.
\newblock Finite functions and the necessary use of large cardinals.
\newblock {\em Ann. of Math.}, 148:803--893, 1998.

\bibitem[Lin95]{al:iuv}
Andrei Linde.
\newblock Self-reproducing inflationary universe.
\newblock {\em Scientific American}, 271(9):48--55, 1995.

\bibitem[RW99]{jg:pos}
Jeffrey~B. Remmel and S.~Gill Williamson.
\newblock Large-scale regularities of lattice embeddings of posets.
\newblock {\em Order}, 16:245--260, 1999.

\bibitem[Teg09]{mt:mhy}
Max Tegmark.
\newblock The multiverse hierarchy.
\newblock {\em arXiv:0905.1283v1 [physics.pop-ph]}, 2009.

\end{thebibliography}

\end{document}